\documentclass{llncs}
\usepackage{graphicx}
\usepackage{amssymb}

\usepackage{amsthm}
\usepackage{amsmath}

\newcommand{\la}{\ensuremath{ \mathcal{L}}}

\usepackage[usenames,dvipsnames,table]{xcolor}
\newcommand{\sd}[1]{{#1}}

\usepackage{todonotes}

\newcommand{\abs}[1]{\left | #1 \right |}

\newcommand{\drop}[1]{}

\newtheorem*{warmup}{Warm-up}

\begin{document}

\title{Dynamic and Multi-functional Labeling Schemes}


\author{S\o ren Dahlgaard, Mathias B\ae k Tejs Knudsen and Noy Rotbart }

\institute{
 Department of Computer Science, University of Copenhagen\\
  Universitetsparken 5, 2100 Copenhagen\\
  \texttt{\{soerend,makn,noyro\}@di.ku.dk} }

\maketitle

\begin{abstract}
We investigate labeling schemes supporting  adjacency, ancestry, sibling,
and connectivity queries  in forests. 
In the course of more than 20 years, the existence of $\log n + O(\log \log)$  labeling schemes supporting each of these  functions was proven, with the most recent being ancestry [Fraigniaud and Korman, STOC '10].
Several multi-functional labeling schemes also enjoy lower or upper bounds of
$\log n + \Omega(\log \log n)$ or $\log n + O(\log \log n)$ respectively.
Notably an upper bound of $\log n + 5\log \log n$ for adjacency+siblings
and a lower bound of $\log n + \log \log n$ for each of the functions siblings, ancestry, and
connectivity [Alstrup et al., SODA '03].
We improve the constants hidden in the $O$-notation. In particular we show a
$\log n + 2\log \log n$ lower bound for connectivity+ancestry 
and connectivity+siblings, as well as an upper bound of
$\log n + 3\log \log n + O(\log \log \log n)$ for connectivity+adjacency+siblings
by altering existing methods.

In the context of dynamic labeling schemes it is known that
ancestry requires $\Omega(n)$ bits [Cohen, et~al.~PODS '02]. 
In contrast, we show upper and lower bounds on the label size for adjacency, siblings, and
connectivity of $2\log n$ bits, and $3 \log n$  to support all three functions.
There exist efficient adjacency labeling schemes for planar, bounded treewidth, bounded arboricity and interval graphs.
In a dynamic setting, we show a lower bound of  $\Omega(n)$ for each of those families.
\end{abstract}

\section{Introduction}\label{sec:intro}

A labeling scheme is a  method of distributing the information about the
structure of a graph among its vertices by assigning short \emph{labels},
such that a selected function on pairs of vertices can be computed using only their labels.
The concept  was introduced in a restricted manner by Bruer and Folkman~\cite{Breuer67}, revisited by Kannan, Naor and Rudich~\cite{Kannan92},
and explored by a wealth of subsequent work~\cite{Alstrup02,Alstrup05,Korman10,alstrup2013near,Thorup01,Peleg00}.

Labeling schemes for trees have been studied extensively in the literature due
to their practical applications in improving the performance of XML search
engines. Indeed, XML documents can be viewed as labeled forests, and typical
queries over the documents amount to testing classic properties such as
adjacency, ancestry, siblings and connectivity  between such labeled tree nodes
\cite{wu2004prime}. In their seminal paper, Kannan et. al.~\cite{Kannan92}
introduced labeling schemes using at most $2\log n$ \footnote{Throughout this
paper we let $\log n = \lceil \log_2 n \rceil$ unless stated otherwise.} bits
for each of the functions adjacency, siblings  and ancestry. Improving
these results have been motivated heavily by the fact that a small improvement
of the label size may contribute significantly to the performance of XML search
engines. Alstrup, Bille and Rauhe~\cite{Alstrup05} established a lower bound of
$\log n + \log \log n$ for the functions siblings, connectivity and ancestry
along with a matching upper bound for the first two. \sd{For adjacency, a $\log n +
O(\log^* n)$ labeling scheme was presented in \cite{Alstrup02}.} A $\log n +
O(\log \log n)$  labeling scheme for ancestry was established only recently by
Fraigniaud and Korman~\cite{Korman10}.

In most settings, it is the case that the \sd{structure of the graph to be
labeled}\drop{data structure} is not known in advance.
In contrast to the \emph{static} setting described above, a \emph{dynamic}
labeling scheme typically receives the tree as an online sequence of
\sd{topological events}\drop{addition of leaves}, with a natural extension
that includes  removal of
leaves. Cohen, Kaplan and Milo~\cite{cohen2010labeling} considered
\emph{dynamic labeling schemes} where the  encoder receives $n$  leaf
insertions and assigns unique   labels that must remain unchanged throughout
the labeling process. In this context, they showed  a tight bound of
$\Theta(n)$ bits for any dynamic ancestry labeling scheme\drop{\footnote{From
hereon, the lower bound  is referred as Cohen's  bound.}}. We stress the
importance of their lower bound by showing that it extends to  routing, NCA,
and distance as well. In light of this lower bound, Korman, Peleg and
Rodeh~\cite{korman2004labeling} introduced  dynamic labeling schemes, where
node re-label is permitted and performed by message passing. In this model they
obtain a compact labeling scheme for ancestry, while keeping the number of
messages small. Additional results in this setting include conversion methods
for static labeling schemes~\cite{korman2004labeling,korman2007general}, as
well as specialized distance~\cite{korman2007general,korman2007labeling} and
routing~\cite{korman2008improved,korman2009compact} labeling schemes.
See~\cite{Rotbart14} for experimental evaluation.

Considering the static setting, a natural question is to determine the label
size required to support some, or all, of the functions. Simply concatenating
the labels mentioned yield a $O(\log n)$ label size, which is clearly
undesired. Labeling schemes supporting multiple functions (or multi-functional
labeling schemes) were previously studied in \cite{Alstrup05}, showing an upper
bound of $\log n + 5\log\log n$ bits for combined adjacency and sibling
queries. \sd{We observe, that their scheme can be combined with the ideas of
\cite{Alstrup02} to produce a $\log + 2\log\log n$ labeling scheme for
adjacency and siblings.}
\drop{We show that most labeling schemes can be altered to also support
connectivity queries by adding an extra $\log\log n$ bits.
 Thm.~\ref{thm:static-combo-lb} shows that such a labeling scheme requires at
 least $\log n + 2 \log \log n$ bits -- even if it supports just connectivity
 and sibling/ancestry queries.}

See Table~\ref{table:complexities} for a summary of labeling schemes for
forests including the results of this paper.

\begin {table}[h]
	\begin{center}
	    \begin{tabular}{ | l | l | l | l | }
		    \hline
		   \textbf{Function } & \textbf{Static Label Size} & \textbf{Static Lower Bound}  & \textbf{Dynamic} \\ \hline\hline
		   Adjacency  	&$\log n +O(\log^* n)$~\cite{Alstrup02}  & $\log n
            +O(1)$ & {\color{blue} $2 \log n$~(Th.~\ref{thm:simple-dyn}}) \\ \hline
		   Connectivity  	&$\log n +\log \log n$~\cite{Alstrup05} & $\log n
            +\log \log n$~\cite{Alstrup05}  &  {\color{blue} $2 \log
        n$~(Th.~\ref{thm:simple-dyn}}) \\ \hline
		   Sibling  		&$\log n +\log \log
        n$~\cite{lewenstein2013succinct} & $\log n +\log \log
        n$~\cite{Alstrup05}  & {\color{blue} $2 \log n$~(Th.~\ref{thm:simple-dyn}}) \\ \hline
		   Ancestry   	&$\log n +4\log \log n$~\cite{Korman10} & $\log n + \log \log n$~\cite{Alstrup05} & $n$~\cite{cohen2010labeling} \\ \hline \hline
		  AD/S		& {\color{blue}  $\log n + 2\log \log n$~(Cor.~\ref{cor:adj_sib_static})}& 	$\log n + \log
        \log n$~\cite{Alstrup05}& {\color{blue} $2 \log n$~(Th.~\ref{thm:simple-dyn})} \\ \hline
		   C/S  		& {\color{blue} $\log n +2\log \log
    n$~(Th.~\ref{thm:static-con-alter-ub})} & $\log n +2\log \log
        n$~(Th.~\ref{thm:static-combo-lb}) & {\color{blue} $3 \log
    n$~(Th.~\ref{thm:combo_dyn})} \\ \hline
		   C/AN   		&{\color{blue} $\log n +5\log \log
n$~(Th.~\ref{thm:static-con-alter-ub})} & {\color{blue} $\log n + 2\log \log n$~(Th.~\ref{thm:con-anc})} &  $n$~\cite{cohen2010labeling} \\ \hline
		   C/AD/S   	& {\color{blue} $\log n + 3\log \log
    n$~(Cor.~\ref{cor:adj_sib_static})} & {\color{blue} $\log n + 2\log \log
    n$~(Th.~\ref{thm:static-combo-lb})} & {\color{blue} $3 \log
    n$~(Th.~\ref{thm:combo_dyn})} \\ \hline	\hline
		   Routing   	&$(1+o(1))\log n$~\cite{Thorup01}   &$\log n + \log \log n$~\cite{Alstrup05} & {\color{blue} $n$~(Sec.~\ref{Sec:Dynamic})}  \\  \hline
		   NCA   		&$2.772 \log n$~\cite{alstrup2013near}  & $1.008 \log n$~\cite{alstrup2013near}  &  {\color{blue} $n$~(Sec.~\ref{Sec:Dynamic})} \\ \hline
		    Distance	& $1/2 \log^2 n$~\cite{Peleg00}  & $1/8 \log^2 n$~\cite{Peleg00}  & {\color{blue} $n$~(Sec.~\ref{Sec:Dynamic})} \\  \hline\hline
		    Sibling*  		&$\log n$  & $\log n$  & $ \log n$ \\ \hline
		   Connectivity*  	&$\log n$ & $\log n $  & $ \log n$ \\ \hline
		    C/S*  	& {\color{blue} $\log n+ \log \log n
$~(Th.~\ref{thm:static-con-alter-ub})} & {\color{blue} $\log n+ \log \log
n$~(Th.~\ref{thm:con_sib_dyn_lb})}  & $ 2\log n$ \\ \hline
	    \end{tabular}
	 \end{center}
	 	 	\caption{Upper and lower  label sizes  for labeling  trees with $n$ nodes (excluding additive constants).
			 Routing  is reported in the designer-port model~\cite{Fraigniaud01} and NCA with no pre-existing labels~\cite{alstrup2013near},
			 functions marked with * denote non-unique labeling schemes, and bounds without a reference are folklore.
			 Dynamic labeling schemes are all tight.}
	\label{table:complexities}
\end {table}

\subsection{Our contribution}
We first \sd{observe} that for \sd{the dynamic setting, we can achieve efficient
labeling schemes for the functions adjacency, sibling, and connectivity
without the need of relabeling.}
\drop{the functions adjacency, sibling and connectivity no
relabeling is required to achieve efficient dynamic labeling schemes.}
More precisely, we observe that the original $2 \log n$ adjacency labeling
scheme due to  Kannan et. al.~\cite{Kannan92} is in fact suitable for the
dynamic setting. Moreover, the original labeling scheme also supports sibling
queries and a slightly modified scheme is shown to work for connectivity. We
also present simple families of insertion sequences for which labels of size
$2\log n$ are required, showing that in the dynamic setting the original
labeling schemes are in fact optimal. The result is in contrast to the static
case, where adjacency labels requires strictly fewer bits than both sibling and
connectivity. The labeling schemes also reveal an exponential gap between
ancestry and  the functions mentioned for the dynamic setting. In
Section~\ref{sec-other-graphs} we show a construction of simple lower bounds of
$\Omega(n)$ for adjacency labeling schemes on various important graph
families.

In the context of multi-functional labeling schemes, we show the following
results. First, we show that $3 \log n$ bits are necessary and sufficient for
any dynamic labeling scheme supporting adjacency and connectivity. Turning to
static labeling schemes, we show a tight $\log n + 2 \log \log n$ bound for
any unique labeling scheme supporting both connectivity and siblings/ancestry.
\sd{For the upper bound, we prove the more general result, that}
\drop{In order to show the upper bound, we prove that} any labeling scheme of size
$S(n)$ growing faster than $\log n$ can be altered to support connectivity as well by adding at most $\log\log n$ bits.
\sd{Coupled with our observation, that \cite{Alstrup02} and \cite{Alstrup05}
provide a $\log n + 2\log\log n$ scheme for adjacency and sibling, this
provides a $\log n + 3\log\log n$ labeling scheme for all the functions
adjacency, sibling and connectivity.}
\section{Preliminaries}
A binary string $x$ is a member of the set $\{0,1\}^*$, and we denote its size by $\vert x \vert$, and the concatenation of two binary strings $x,y$ by $ x \circ y$.

A \emph{label assignment}  for a tree $T$ is a mapping of each $v \in V$
	to a bit  string $\la(v)$, called  the \emph{label} of $v$.
Given  a tree $T=(V,E) $ rooted in $r$ with $n$ nodes, and let $u,v \in V$.
	The function $adjacency(v,u)$ returns \textbf{true} if and only if $u$ and $v$ are adjacent in $T$,
	$ancestry(v,u)$   returns \textbf{true} if and only if $u$ is on the path $r \leadsto v$,
	$siblings(v,u)$   returns \textbf{true} if and only if $u$ and $v$ have the same parent in $T$\footnote{By this definition, a node is a sibling to itself.},
	$routing(v,u)$   returns an identifier  of the edge connected to  $u$ on the path to  $v$,
	$NCA (v,u)$   returns the label of the first node in common on the paths $u \leadsto r$ and $v \leadsto r$,
	and $distance (v,u)$ returns the length of the path from $v$ to $u$.
	The functions mentioned previously are also defined for forests.
	Given a rooted forest $F$ with $n$ nodes, for any two nodes $u,v$ in $F$ the function  $connectivity(v,u)$ returns  \textbf{true}  if $v$ and $u$ are in the same tree in $F$.

	Given a function $f$ defined on sets of vertices,
	 an  \emph{ f-labeling scheme} for a family of graphs $\mathcal{G}$ consists   of  an encoder  and decoder.
	The  \emph{encoder}  is an algorithm that receives a graph $G \in
    \mathcal{G}$ as input and  computes a label assignment $e_G$. If the encoder receives $G$ as a sequence of topological events\footnote{Cohen et al. defines such a sequence as  a set of insertion of nodes into an initially empty tree, where the root is inserted first,and all other insertions are of the form ``insert node $u$ as a child of node $v$''.  We extend it to support ``remove leaf $u$'', where the root may never be deleted. } the labeling scheme is \emph{dynamic}.
	 The \emph{decoder}  is an algorithm  that receives any two labels $\la(v),\la(u)$  and  computes the query $d(\la(v),\la(u))$, such that  $d(\la(v),\la(u))=f(v,u)$. The \emph{size} of the labeling scheme is the maximum label size.
	 If for all graphs  $ G \in \mathcal{G}$, the label assignment  $e_G$ is an injective mapping, i.e. for all distinct $u,v \in V(G)$, $e_G(u) \neq e_G(v) $, we say that the labeling scheme assigns \emph{unique} labels. Unless stated otherwise, the labeling schemes presented are assumed to assign unique labels. Moreover, we allow the decoder to know the label size.

	 Let $H$ be a family of graphs, a graph $G \in H$, and suppose that an f-labeling scheme assigns a node $v \in G$ the label $\la(v)$.
	  If $\la(v)$ does not appear in any of the label assignments for the other graphs in $H$, we say that the label is \emph{distinct} for the labeling scheme over $H$.
All labeling schemes constructed in this paper require $O(n)$ encoding time and $O(1)$ decoding time under the assumption of a $\Omega(n)$ word size RAM model. See~\cite{Thorup01} for additional details.

\section{Dynamic labeling schemes}~\label{Sec:Dynamic}
We first note that the lower bound for ancestry due to Cohen, et.~al.~also
holds for NCA, since the labels computed by an NCA labeling scheme can decide ancestry:
Given the labels  $\la(u),\la(v)$ of two nodes $u,v$ in the tree $T$,  return
true if $\la(u)$ is equal to the label returned by the original NCA decoder, and
false otherwise.
Similarly, suppose a labeling scheme for routing\footnote{Routing in the
designer port model~\cite{Fraigniaud01}.} assigns $0$ as the port number on the
path to the root. Given  $\la(u),\la(v)$  as before, return true if
$routing(\la(u),\la(v)) \neq 0$ and $ routing(\la(v), \la(u)) =0$.
\drop{If there were to exist a dynamic labeling scheme for routing or NCA with
size $o(n)$, the labels produced would be sufficient to determine ancestry,
in contrast to Cohen's bound.}
Peleg~\cite{Peleg05} proved that any $f(n)$ distance labeling scheme can be
converted to $f(n)+\log (n)$ labeling scheme for NCA by  attaching the
depth of any node. Since the depth of a node inserted can not change in our
dynamic setting, we conclude that any lower
bound for ancestry also applies to distance, routing, and NCA.

\subsection{ Upper Bounds} \label{sec:upper-bounds-dynamic}
The following (static) adjacency labeling scheme was introduced by Kannan et al.~\cite{Kannan92}.
Consider an arbitrary rooted tree $T$ with $n$ nodes.
Enumerate the nodes in the tree with the numbers $0$ through $n-1$, and let,
for each node $v$, $Id(v)$ be the number associated with $v$. Let $parent(v)$ be the parent of a node $v$ in the tree.
The label of $v$ is $\la(v) = (Id(v),Id(parent(v)))$, and the root is labeled
$(0,0)$. Given the labels $\la(v),\la(v')$ of two nodes $v$ and $v'$, observe
that the two nodes are adjacent if and only if either $Id(parent(v))=Id(v')$ or
$Id(parent(v'))=Id(v)$ but not both, so that the root is not adjacent to
itself.

This is also  a dynamic labeling scheme for adjacency with equal label size.
 Moreover, it is also both a static and dynamic labeling scheme for sibling,
in which case, the decoder must check if $Id(parent(v))= Id(parent(v'))$.
 A labeling scheme for connectivity can be constructed by storing the
 component number rather than the parent id. After $n$ insertions, each label contains two parts, each in the range
 $[0,n-1]$.  Therefore, the label size required is $2 \log n$.

The labeling schemes suggested extend to larger families of graphs.
In particular, the dynamic connectivity labeling scheme  holds  for the family of all graphs.
The family of $k$-bounded degree graphs enjoys a similar dynamic adjacency labeling scheme of size $(k+1) \log n$.

\subsection{ Lower Bounds}
We  show that $2 \log n$ is in fact a tight bound for any dynamic adjacency labeling scheme for trees.
We denote  by $\mathcal{F}_n(k)$ an  insertion sequence   of  $n$ nodes, creating  an \emph{initial  path} of length $1 < k \leq n $, followed by  $n-k$  \emph{adjacent leaves}  to node $k-1$ on the path.
The family of all such insertions sequences is denoted $\mathcal{F}_n$. For illustration see Fig.~\ref{fig:AdjLowerBound}.

\begin{lemma}\label{lem:pers-anc}
Fix some dynamic labeling scheme that supports adjacency. For any $1 < k < n$, $\mathcal{F}_n(k)$ must contain at least $n-k$ distinct labels \drop{wrt.} for this labeling scheme over $\mathcal{F}_n$.
\end{lemma}

\begin{proof}
The labels of $\mathcal{F}_n(n)$ are set to  $P_1 \dots P_n$ respectively.
Since the encoder is deterministic, and since every insertion sequence
$\mathcal{F}_n(k)$ first  inserts nodes on the initial  path,these nodes must be labeled
$P_1 \dots P_k$. Let the labels of the adjacent leaves of such an insertion
sequence be denoted by $L^k_1 \dots L^k_{n-k}$.

			\begin{figure} [h]
				\centering
				\includegraphics[width=.3\textwidth]{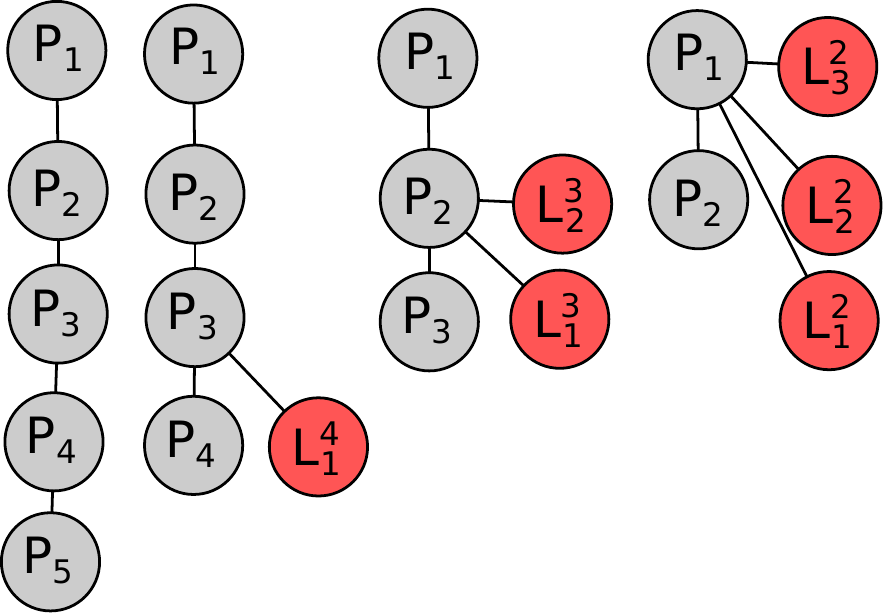}
				\caption{Illustration of $\mathcal{F}_5$. }
				\label{fig:AdjLowerBound}
			\end{figure}

Clearly, $L_1^{k}\dots L_{n-k}^{k}$ must be different from $P_1\dots P_n$, as the only other labels adjacent to $P_{k-1}$ are $P_{k-2}$ and $P_{k}$, which have already been used on the initial path.
 Consider now any node labeled $L_i^j$ of $\mathcal{F}_n(j)$ for $j\ne k$. Assume w.l.o.g that $j>k$. Such a node must be adjacent to $P_{j-1}$ and
 \emph{not} to $P_{k-1}$, as $P_{k-1}$ is contained in the path to $P_{j-1}$. Therefore we must have $L_i^j \notin \{L_1^{k}, \dots, L_{n-k}^{k}\}$.
\end{proof}

Identical lower bounds exist for both sibling and connectivity, see App.~\ref{lower-simple}.
\begin{theorem}\label{thm:simple-dyn}
Any dynamic labeling scheme supporting either  adjacency, connectivity, or sibling  requires at least $ 2 \log n -1 $ bits.
\end{theorem}
\begin{proof}
According to Lem.~\ref{lem:pers-anc}, at least $n+ \sum_{i=2}^{n-1} i =  n^2/2
+O(n)$ distinct labels are required to label $\mathcal{F}_n$ if adjacency or
sibling requests are supported, and the same applies for $\mathcal{F}^c_n$  if
connectivity is supported.
\end{proof}

A natural question is whether a randomized labeling scheme could provide labels of size less than $2 \log n - O(1)$.
The next theorem, based on Theorem~3.4 in \cite{cohen2010labeling} answer this question negatively. The proof is deferred to Appendix~\ref{proof:random}.
\begin{theorem}\label{thm:random}
	For any randomized dynamic labelling scheme supporting either
	adjacency, connectivity, or sibling queries
	there exists an insertion sequence such that the expected value
	of the maximal label size is at least $2 \log n - O(1)$ bits.
\end{theorem}

\subsection{Other Graph Families}\label{sec-other-graphs}
In this section, we expand our lower bound ideas to adjacency labeling schemes for the  following families: bounded arboricity-$k$ graphs\footnote{The \emph{arboricity} of a graph $G$  is the minimum number of edge-disjoint acyclic subgraphs whose union is $G$.} $\mathcal{A}_k$,  bounded degree-$k$ graphs $\Delta_k$, and  bounded treewidth-$k$ graphs $\mathcal{T}_k$.

In the context of (static) adjacency  labeling schemes,  these families  are   well studied~\cite{Kannan92,Alstrup02,gavoille2007shorter,chung1990universal,Adjiasvhili14}
In particular,  $\mathcal{T}_k$, $\Delta_k$ and  $\mathcal{A}_k$  enjoy  adjacency labeling schemes of size $\log \log (n/k))$ \cite{gavoille2007shorter}, and $k \log n+O(\log^* n)$ \cite{Alstrup02} respectfully.

We consider a sequence of node insertions along with all edges adjacent to them, such that an edge $(u,v)$ may be introduced along with node $v$ if  node $u$ appeared prior in the sequence, and prove the following.
\begin{theorem}\label{theorem:lower-arboricity}
 Any dynamic adjacency labeling scheme for  $\mathcal{A}_2$ requires $\Omega(n)$ bits.
\end{theorem}
\begin{proof}
Let $S$ be the collection of all  nonempty  subsets  of  the integers $1 \dots n-1$.
Since there are $2^{n-1}-1$ such sets possible, $\vert S \vert = 2^{n-1}-1$.
For every $s \in S$, we denote  by $\mathcal{F}_n(s)$ an  insertion sequence   of  $n$ nodes, creating  a path of length $n-1$, followed by  a single node $v$ connected to the nodes on the path whose number is a member of $s$.
Such a graph has  arboricity $2$ since it can be decomposed into an initial path and a star rooted in $v$.
For each of the $|S|$  insertion sequences, the label of $v$ must be distinct.
We  conclude that the number of bits required for any adjacency labeling scheme
is at least $\log(|S|)= n-1$ bits.
See Fig.~\ref{fig:AdjLowerBoundGraphs} for illustration. 
\end{proof}
\vspace*{-5ex}
			\begin{figure}
				\centering
				\includegraphics[width=.25\textwidth]{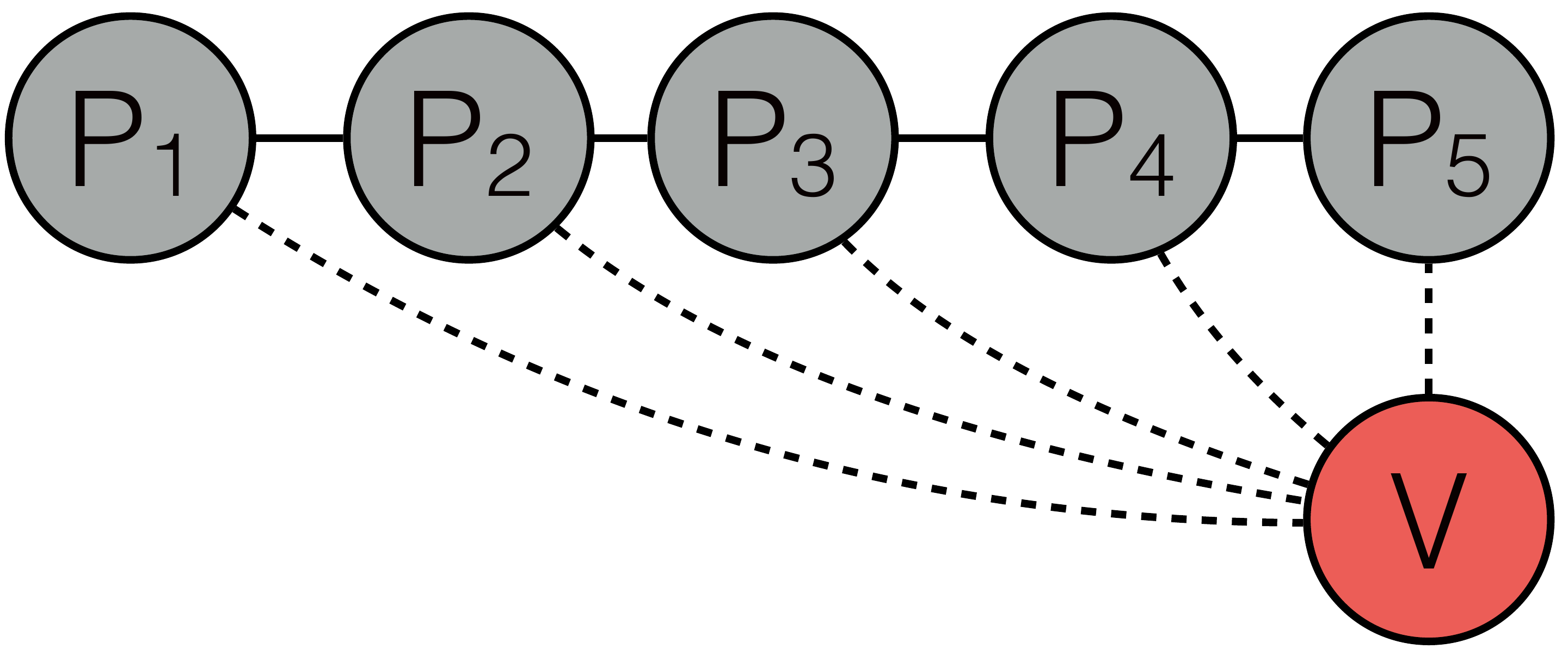}
				\caption{Illustration of $\mathcal{F}(s)$ for $n=5$. The dotted lines may or may not appear in the insertion sequence depending on the element of $S$ chosen. }
				\label{fig:AdjLowerBoundGraphs}
			\end{figure}
\vspace*{-2ex}
The construction of  $\mathcal{F}_n(s)$ implies an identical lower
bound for the family of planar graphs, as well as interval graphs. By
 considering all sets $s$ of at most $k$ elements instead, we get
a bound of $k \log n$ label size for any adjacency labeling scheme for $\Delta_k$, where $k$ is constant.

To show a similar bound on $\mathcal{T}_k$, we prove  that the sequence of insertions creates graphs in $\mathcal{T}_3$.
For every face $R$ in  a planar embedding $M$ of a  planar graph $G$, define
$g(R)$ to be the minimum value of $k$, such that there is a sequence of faces
$R_0 \dots R_k$, with $R_0$ the exterior face, and $R_k = R$, and for $1 \leq j
\leq k$, there is a vertex $v$ that is both on face $R_{j-1}$ and $R_{j}$. The
radius of $M$ is the minimum value of $g$ such that $g(R) \leq g$ for all
regions $R$ of $M$.
\begin{lemma}\cite{bodlaender1988dynamic}\label{lemma:bodlander}
Let $G=(V,E)$ be a planar graph with radius $\leq g$, $ g \geq 1$, then $G$ has treewidth at most $3d$.
\end{lemma}
The lemma is useful for our purposes since  the graphs in the  family of  planar graphs resulting from  $\mathcal{F}(s)$ have radius $1$.
\begin{corollary}
 Any dynamic adjacency labeling scheme for $\mathcal{T}_k$, where  $k \geq 3$, requires $\Omega(n)$ bits.
\end{corollary}

\section{Multi-Functional Labeling schemes}
In this section we investigate labeling schemes incorporating  two or more of the functions mentioned.
\subsection{Dynamic Multi-Functional Labeling Schemes}
A dynamic labeling scheme for answering any combination of connectivity, adjacency and sibling queries at the same time can be obtained by setting  $\la(v) = (Id(v), Id(parent(v)), component(v))$ as described in Section~\ref{sec:upper-bounds-dynamic}  which result in a $3 \log n$ labeling scheme.

 We  now show that this upper bound is in fact is tight.
 More precisely, we show that $3\log n$ bits are required to answer  the combination of connectivity and adjacency.
  Let $I_n(j,k)$ be an insertion sequence designed as follows: First $j$ nodes are inserted creating an \emph{initial
forest} of single node trees. Then $k$ nodes are added as a path with root in
the $j$th tree. At last, $n-j-k$ adjacent \emph{path leaves} are added to the
second-to-last node on the path. For a given $n$ we define $I_n$ as the family
of all such insertion sequences. See Fig.~\ref{fig:combi_lb} for reference.
\begin{figure}[h]
    \centering
    \includegraphics[width=.6\textwidth]{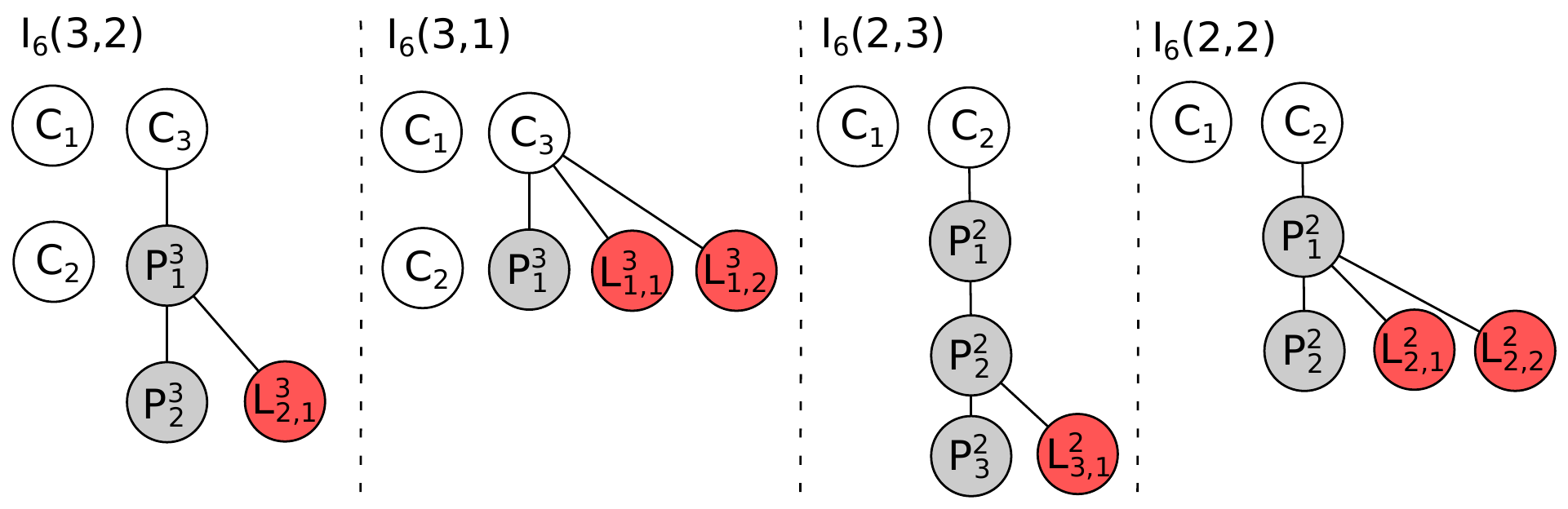}
    \caption{Illustration of $I_n(j,k)$ for specific values of $j$, $k$, and
    $n$.}
    \label{fig:combi_lb}
\end{figure}

\begin{lemma}\label{lem:combi_lb}
   Fix some dynamic labeling scheme that supports adjacency and connectivity requests.
    For any $1 < j+k < n$, $I_n(k)$ must contain at least $n-j-k$ distinct labels for this labeling scheme over $I_n$.
\end{lemma}
The  proof of  Lem.~\ref{lem:combi_lb} is found in  App.~\ref{Lemma-combi-bi}.

\begin{theorem}\label{thm:combo_dyn}
    Any dynamic labeling scheme supporting both adjacency and connectivity
    queries requires at least $3\log n - O(1)$ bits.
\end{theorem}
\begin{proof}
	According to Lem.~\ref{lem:combi_lb} at least
	$\sum_{j=1}^{n-1}\sum_{k=1}^{n-j-1} n-j-k = \frac{1}{6}n^3 - O(n^2)$ distinct labels
	are required to label the family $I_n$. Thus a label size of at least $3\log n
	- O(1)$ bits is needed by any dynamic labeling scheme.
\end{proof}

The same family of insertion sequences can be used to show a $3\log n - O(1)$
lower bound for any dynamic labeling scheme supporting both sibling and
connectivity queries. Furthermore, similarly to Theorem~\ref{thm:random}, the
bound holds even without the assumption that the encoder is deterministic.

\subsection{Static Multi-Functional   labeling schemes}~\label{sec:static-multi}
As seen in Thm.~\ref{thm:combo_dyn}, the requirement to support both connectivity and adjacency force an increased label size for any dynamic labeling scheme.
In this section we prove  lower and upper bounds for static labeling schemes that support those operations, both for the case where the labels are necessarily unique, and for the case that they are not. From hereon, all labeling schemes are on the family of rooted forests with at most $n$ nodes.

\begin{theorem}\label{thm:static-con-alter-ub}
    Consider any function $f$ of two nodes in a single tree. If there exists an
    $f$-labeling scheme of size $S(n)$, where $S(n)$ is non-decreasing and
    $S(a) - S(b) \ge \log a - \log b - O(1)$ for any $a\ge b$.
    Then there exists an $f$-labeling scheme, which also supports connectivity
    queries of size at most $S(n) + \log\log n + O(1)$.
\end{theorem}
\begin{proof}
    We will consider the label $\la(v) = \{C\circ L\circ sep\}$ defined as
    follows. First, sort the trees of the forest according to their sizes. For the $i$th biggest tree we set $C = i$ using $\log i$ bits. Since the tree has at most $n/i$ nodes, we can pick the label $L$ internally in the tree using only $S(n/i)$ bits. Finally, we need a separator, $sep$, to separate $C$ from $L$. We can represent this using $\log\log n$ bits, since $i$ uses at most $\log n$ bits.

    The total label size is this $\log i + S(n/i) + \log\log n + O(1)$ bits, which is less than
    $S(n) + \log\log n + O(1)$ if $S(n) - S(n/i) \ge \log i - c$ for some
    constant $c$, which holds by our assumption. Since $f$ is a function of two nodes from
    the same tree, this altered labeling scheme can answer both queries for $f$
    as well as connectivity. It is now required that any label assigned has size exactly $S(n) + \log\log n$ bits, so that the decoder may correctly identify $sep$ in the bit string. For that purpose we  pad labels with less bits with sufficiently many $0$'s.
\end{proof}

As a special case, we get a labeling scheme for connectivity and
sibling/ancestry for $\log n + 2\log\log n$ and for connectivity and sibling of
$\log n + \log\log n$ if the labels need not be unique.

The following corollary is a direct result of~\cite{Alstrup02,Alstrup05}. A
sketch of the proof is found in App.~\ref{app:adj_sib}.
\begin{corollary}\label{cor:adj_sib_static}
There exists unique labeling scheme supporting both sibling and adjacency  queries of size at most $\log n + 2 \log \log n$.
\end{corollary}

\subsubsection{Lower Bound}
We now show, that the upper bounds implied  by Theorem~\ref{thm:static-con-alter-ub} for labeling schemes supporting siblings and connectivity  are indeed tight for both the unique and non-unique  cases.
To that end we consider the following forests: For any integers $a,b,n$ such that $ab \mid n$
denote by $F_n(a,b)$ a forest consisting of $a$
components (trees), each with $b$ \emph{sibling groups}, where each sibling
group is composed of $\frac{n}{a\cdot b}$ nodes. Note that $F_n(a,b)$ has at least
$n$ but no more than $2n$ nodes.

Our proofs work as follows:
Firstly, for any two forests $F_n(a,b)$ and $F_n(c,d)$ as defined above, we
        establish an upper bound on the number of labels that can be assigned
        to both $F_n(a,b)$ and $F_n(c,d)$.
Secondly, for a carefully chosen family of forests $F_n(a_1,b_1), \ldots,
        F_n(a_k,b_k)$, we  show that when labeling $F_n(a_i,b_i)$ at least
        a constant fraction of the labels has to be distinct from the labels of
        $F_n(a_1,b_1),\ldots, F_n(a_{i-1},b_{i-1})$.
Finally,  by summing over each $F_n(a_i, b_i)$ we show that a sufficiently
        large number of bits are required by any labeling scheme supporting the
        desired queries.
        
Our technique is a simpler version of the boxes and groups argument of Alstrup et~al.~\cite{Alstrup05}, and generalizes to the case of two nested equivalence classes, namely connectivity and siblings. The proofs for Lem.~\ref{lem:forest-reuse} and~\ref{lem:set-reuse-bound} are in 
App.~\ref{proof:forest-reuse} and App.~\ref{proof:set-reuse-bound} respectively.

\begin{lemma}\label{lem:forest-reuse}
    Let $F_n(a,b)$ and $F_n(c,d)$ be two forests such that $ab \ge cd$.
    Fix some unique labeling scheme supporting both
    connectivity and siblings, and denote the set of labels assigned to
    $F_n(a,b)$ and $F_n(c,d)$ as $e_1$ and $e_2$ respectively.
	Then
    \[
        |e_1\cap e_2| \le \min(a,c)\cdot\min(b,d)\cdot\frac{n}{a\cdot b}\ .
    \]
\end{lemma}

\begin{lemma}\label{lem:set-reuse-bound}
    Let $F_n(a_1,b_1),\ldots, F_n(a_i,b_i)$ be a family of forests with
    $a_1\cdot b_1\le \ldots\le a_i\cdot b_i$. Assume there exists a unique
    labeling scheme supporting both connectivity and siblings, and let $e_j$
    denote the set of labels assigned by such a scheme to the forest
    $F_n(a_j,b_j)$. Assume that the sets $e_1,\ldots,e_{i-1}$ have already
    been assigned. Then the number of distinct labels the encoder must
    introduce when assigning $e_i$ is at least
    \[
        n - \sum_{j=1}^{i-1} \min(a_j,a_i)\cdot \min(b_j,b_i)\cdot
        \frac{n}{a_i\cdot b_i}\ .
    \]
\end{lemma}

We now use Lem.~\ref{lem:set-reuse-bound} to show the following
known result~\cite{Alstrup05}.

\begin{warmup}
    Any static labeling scheme for connectivity queries requires at least $\log
    n + \log\log n - O(1)$ bits.
\end{warmup}
\begin{proof}
    Consider the family of $\log_3 n$ forests $F_n(1,1), F_n(3,1), \ldots,
    F_n(\log_3 n, 1)$. Since no two nodes are siblings we can use this forest
    combined with Lem.~\ref{lem:set-reuse-bound} as a lower bound for
    connectivity. Let $e_j$ denote the label set assigned by an encoder for
    $F_n(3^j,1)$. We assume that the labels are assigned in the order
    $e_0,\ldots,e_{\log_3 n}$. By Lem.~\ref{lem:set-reuse-bound} the number of
    distinct labels introduced when assigning $e_j$ is at least
    \[
        n - n\sum_{i=0}^{j-1} 3^{i-j} > n/2\ .
    \]
    It follows that labeling the $\log_3 n$ forests in the family requires at
    least $\Omega(n\log n)$ distinct labels.
\end{proof}

This idea extends to some cases of non-unique labeling schemes, as seen in
the theorem below. The proof  of Thm.~\ref{thm:con_sib_dyn_lb} is included in App.~\ref{proof:con_sib_dyn_lb}.

\begin{theorem}\label{thm:con_sib_dyn_lb}
    Any static labeling scheme supporting both connectivity and sibling queries
    requires at least $\log n + \log\log n - O(1)$ bits if the labels need not
    be unique.
\end{theorem}

\begin{theorem}\label{thm:static-combo-lb}
    Any unique static labeling scheme supporting both connectivity and sibling queries
    requires labels of size at  least  $\log n + 2\log\log n - O(1)$.
\end{theorem}
\begin{proof}
	Fix some integer $x$, and assume that $n$ is a power of $x$.
    We consider the family of forests $F_n(1,1), F_n(x,1), F_n(1,x),
    F_n(x^2,1),$ $F_n(x,x), F_n(1,x^2), \ldots,$\linebreak $F_n(1,x^{\log_x n})$.

    Let $e_a^b$ denote the label set assigned to $F_n(x^a,x^b)$ by an
        encoder. We  assign
        the labels in the order $e_0^0, e_1^0, e_0^1, e_2^0, e_1^1, \ldots,
        e_0^{\log_x n}$. Thus, when assigning $e_a^b$ we have already assigned
        all label sets $e_c^d$ with $c+d < a+b$ or $c+d = a+b$ and $d < b$. By
        Lem.~\ref{lem:set-reuse-bound}, the number of distinct labels introduced
        when assigning $e_a^b$ is at least
        \[
            n - \sum_{\substack{c+d<a+b\\c,d\ge 0}}
            \frac{n}{x^{a+b}}\cdot x^{\min(a,c)+\min(b,d)} +
            \sum_{d=0}^{b-1} \frac{n}{x^{a+b}}\cdot x^{a+d} \\
        \]
        This counting argument is better demonstrated in
        Fig.~\ref{fig:combo_static_lb}. In the figure, we are concerned with
        assigning the labels in $e_2^2$. The grey boxes represent the label
        sets already assigned, and the right-side figure shows the fractions of
        $n$ that each set $e_c^d$ at most has in common with $e_2^2$. Observe
        that we can split the above sum into three cases as demonstrated in the
        figure: If $c\le a$ and $d\le b$ the bound
        supplied by Lem.~\ref{lem:forest-reuse} is $x^{c+d-a-b}$. Otherwise,
        either $c>a$ or $d>b$, but not both. If $c>a$, recall
        that $d<b$ so the bound is $x^{d-b}$. For $d>b$ the bound is $x^{c-a}$
        by the same argument. Applying these rules, we see that the number of
        distinct labels introduced is at least
        \begin{align*}
        &n - n\cdot\left(\sum_{c=0}^a \sum_{d=0}^b x^{c+d-a-b} + \sum_{d=0}^{b-1}
        (b-d)\cdot x^{d-b} + \sum_{c=0}^{a-2} (a-c)\cdot x^{c-a}\right) + n \\
        \ge\ &n - n\cdot\left(\frac{x^2 + x + 2}{(x-1)^2}\right)+n  = n - n\cdot \frac{3x+1}{(x-1)^2}\ .
        \end{align*}
        Note that we add $n$, as we have also subtracted $n$ labels for the
        case when $(c,d) = (a,b)$.

    By setting $x = 6$ we get that the encoder must introduce $6n/25$ distinct
    labels for each $e_a^b$. Since we have $\Theta(\log^2 n)$ forests, a total
    of $\Omega(n\log^2 n)$ labels are required for labeling the family of forests.
    Each forest consists of no more than $2n$ nodes, which concludes the proof.
\end{proof}
\vspace{-7ex}
\begin{figure}[htbp]
    \centering
    \includegraphics[width=0.5\textwidth]{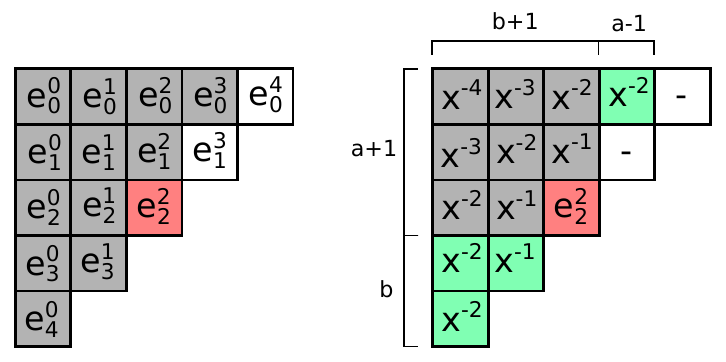}
    \caption{Demonstration of the label counting for $e_2^2$.}
    \label{fig:combo_static_lb}
\end{figure}

The same proof technique is used to prove the following theorem. For completeness, the proof is presented in Appendix~\ref{proof:con-anc}.
\begin{theorem}\label{thm:con-anc}
Any unique static labeling scheme supporting both connectivity and ancestry queries
requires labels of size at least $\log n + 2\log\log n - O(1)$.
\end{theorem}
\section{Concluding remarks}
We have considered multi-functional labels for the functions adjacency, siblings and connectivity.
We also provided a lower bound for ancestry and connectivity.
A major open question is wether it is possible to have a label of size $\log n +O (\log \log n)$ supporting all of the functions.
It seems unlikely that the best known labeling scheme for ancestry~\cite{Korman10} can be combined with the ideas of this paper.

In the context of dynamic labeling schemes, if arbitrary node insertion is permitted, neither adjacency nor sibling labels are  possible. All dynamic labeling schemes also operate when removal is allowed, simply by erasing the label to be removed. Moreover, if the tree contains  not more than $n$ nodes at any moment, it is easy to show that labels of size  2 $\log n$  are necessary and sufficient for each of the  functions.
  \bibliographystyle{elsarticle-num}
  \bibliography{Bibliography}
\appendix
\section{Missing proofs}
\subsection{Lower bound for  dynamic labeling schemes}\label{lower-simple}
For the function sibling we use  the same family and a slightly different argument as follows.
First, it again holds that  $L_1^{k}\dots L_{n-k}^{k}$ must be different from
$P_1\dots P_n$, as they are the  only nodes that are siblings to $P_{k}$.
Furthermore, in $F_n(j)$ the label $L_i^j$ (where $j>k$) is not a sibling of
$P_{k}$, so $L_i^j$ must be distinct from $\{L_1^{k}, \dots, L_{n-k}^{k}\}$.

Finally, for an identical  lower bound on connectivity we define
$\mathcal{F}^c_n(k)$ to be an  insertion sequence   of  $n$ nodes, creating  an
\emph{initial  forest } of  $1 < k < n $ single node trees, followed by  $n-k$
leaves adjacent to tree $k-1$. 

\subsection{Proof of Theorem~\ref{thm:random}} \label{proof:random}
\newcommand{\F}{\mathcal{F}}
We  prove the theorem for labeling schemes supporting adjacency requests.
The proof is similar for the two other types of labeling schemes.
Consider the set $F_n = \{\F_n(k) \mid 1 < k < n/2\}$ consisting of $\Theta(n)$
different insertion sequences, and say that we uniformly choose an insertions
sequence $S \in F_n$.
Fix a \emph{deterministic} labeling scheme supporting adjacency requests.
Each of $\F_n(k) \in F$ has $n-k > \frac{n}{2}$
labels which are distinct for this labeling scheme over $F_n$ (by
Lem.~\ref{lem:pers-anc}).
Say that we write $F_n$ as $F_n = \{S_1, S_2,\ldots, S_{\abs{F_n}}\}$ such that
the maximal label size of the distinct labels over $F_n$ from $S_i$
is smaller than that from $S_j$ if $i < j$. Now consider all the labels from the
insertion sequences $S_1,\ldots,S_i$ which are distinct over $F_n$. There are at
least $\frac{in}{2}$ of those meaning that at least one has label size $\log(in/2)$.
This means that there is a label from $S_i$ which is distinct over $F_n$ and
has label size $\ge \log n + \log i - 1$. This means that the expected value of
the maximal label size of $S$ (which is uniformly drawn from $F_n$) is at least:
\begin{align*}
\frac{1}{\abs{F_n}} \sum_{i=1}^{\abs{F_n}}
\left ( \log n + \log i - 1 \right )
& =
(\log n - 1) +
\frac{1}{\abs{F_n}}
\left (
    \abs{F_n}\log(\abs{F_n}) - O(\abs{F_n})
\right )
\\
& =
\log n + \log \abs{F_n} - O(1)
=
2\log n - O(1)
\end{align*}
Since this holds for any deterministic algorithm Yao's principle yields that
for any randomized algorithm there exists $\F_n(k) \in F_n$ such that the
expected value of the maximal label size is at least $2\log n - O(1)$ on that
insertion sequence.

\subsection{Proof of Lemma~\ref{lem:combi_lb}} \label{Lemma-combi-bi}
    Let $C_1,\ldots, C_n$ be the labels of $I_n(n,0)$ and let $P^j_1,\ldots,
    P^j_{n-j}$ be the labels of the path created by the insertion sequence
    $I_n(j,n-j)$. Since the encoder is deterministic, any insertion sequence
    $I_n(j,k)$ must assign the labels $C_1,\ldots, C_j$
    and $P^j_1,\ldots, P^j_k$ to the first $j+k$ nodes.

    Let $L^j_{k,i}$ denote the label of the $i$th \emph{path leaf} added as a
    part of the insertion sequence $I_n(j,k)$. Clearly
    $L^j_{k,i}$ is different from any $C_{j'}$ and $P^{j'}_{k'}$ by the
    argument of the proof of Lem.~\ref{lem:pers-anc}.

    Consider now two different leaves labeled $L^j_{k,i}$ and $L^{j'}_{k',i'}$.
    If $j=j'$ and $k=k'$ the labels must be different, as they are part of the
    same insertion sequence.

    If $j < j'$ then by looking at $I_n(j,k)$, $L^j_{k,i}$ and $C_j$ are connected.
    By looking at $I_n(j',k')$, $L^{j'}_{k',i'}$ and $C_{j}$ are not connected. Hence
    the labels are different. The case $j > j'$ is symmetric.
    If $j = j'$ and $k < k'$ then by looking at $I_n(j,k)$, $L^j_{k,i}$ and $P_{k}^j$
    are adjacent. And from $I_n(j',k')$ we see that $L^{j'}_{k',i'}$ and $P_{k}^j$ are
    not adjacent. Hence the labels are different. The case $k > k'$ is symmetric.

    In conclusion no two leaves get the same label in any of $I_n(j,k)$.	Since $I_n(j,k)$
    has $n-j-k$ leaves this means that $I_n(j,k)$ contains $n-j-k$ labels that are distinct
    for the labelling scheme over $I_n$.

\subsection{Proof sketch for Corollary \ref{cor:adj_sib_static}}
\label{app:adj_sib}

It was shown in \cite{Alstrup02} how to create a labeling scheme
using a recursive cluster decomposition to support adjacency in $\log n +
O(\log^* n)$ bits. We argue that this decomposition can be combined directly
with the $1$-relationship scheme of \cite{Alstrup05} to create a labeling
scheme supporting both adjacency and sibling using $\log n + 2\log\log n +
O(\log\log\log n)$ bits.

In this proof sketch, we assume that the reader is
familiar with the notations and definitions of \cite{Alstrup02,Alstrup05}.

For $1$-relationship, the scheme of \cite{Alstrup05} actually works with
$\log n + 3\log\log n + O(1)$ bits by storing $spre(parent(v))$ for heavy
nodes instead of only storing $spre(parent(v))$ for light nodes.
The key is to change Lem.~4 in \cite{Alstrup05} to work for heavy nodes. This
is done by considering $pre(v) - 1$ instead of $pre(v)$ for heavy nodes
in the proof. Since
$pre(v) = spre(v)$ we can get label size $\log n + 2\log\log n+O(1)$ for leaves by
adding an extra flag.

The cluster decomposition used in \cite{Alstrup02} works as follows: For some
integer $x$, the tree $T$ is split into $O(n/x)$ clusters of size $O(x)$. Each
cluster has at most two boundary nodes, which are part of more than one
cluster. We can view the clusters as a macro tree, where the nodes are the
boundary nodes and the edges are the clusters. Each cluster is one of three
types (see Fig.~\ref{fig:cluster_ex}): Either it is a leaf cluster
with just one boundary node ($\alpha$), it is a single edge ($\beta$), or it is
an internal cluster with two boundary nodes ($\gamma$). Note that for
$\gamma$-clusters, the top boundary node, $u$, has at most one child inside the
cluster.

\begin{figure}[htbp]
    \centering
    \includegraphics[width=.4\textwidth]{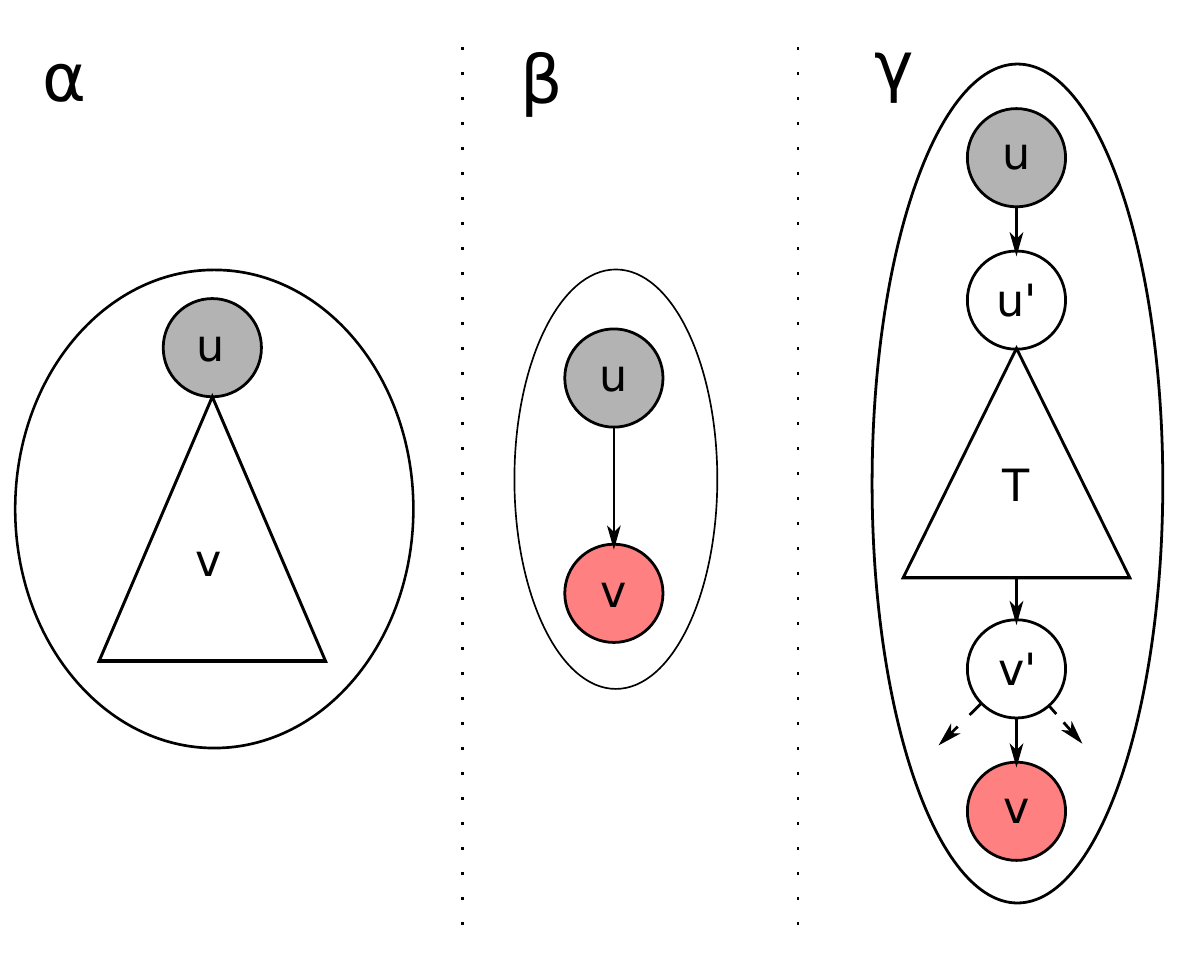}
    \caption{The three different types of clusters.}
    \label{fig:cluster_ex}
\end{figure}

The labeling scheme works by first labeling the macro tree with the modified
$1$-relationship scheme, such that the label of a cluster $C$ is denoted
$\la^M(C)$. Inside each cluster the nodes are labeled, such that the label of a
node $v$ is denoted by $\la^C(v)$.

A node $v$ of the original tree $T$ will be labeled the following way (refer to
Fig.~\ref{fig:cluster_ex} for the node types). Note that upper boundary nodes
$u$ are not included in the cluster -- only lower boundary nodes.
\begin{description}
    \item [Type-$v$ node in $\alpha$-cluster $C$:] We set $\la(v) =
        \{\la^M(C)\circ\la^C(v)\circ type\}$.
    \item [Type-$v$ node in $\beta$-cluster $C$:] We set $\la(v) =
        \{\la^M(C)\circ type\}$.
    \item [Type-$u'$ and type-$v$ nodes in $\gamma$-cluster $C$:] We set
        $\la(u') = \{\la^M(C)\circ type\}$ (and identical for $v$).
    \item [Type $T$ and type-$v'$ nodes in $\gamma$-cluster $C$:] We set
        $\la(v') = \{pre^M(C)\circ\la^C(v')\circ type\}$.
\end{description}

The $type$ parameter is a constant number of bits specifying the following:
Which cluster type is it $\{\alpha,\beta,\gamma\}$. Which type of node is it
$\{$child of $u$ in $\alpha$, type $u'$ in $\gamma$, type $v$ in $\gamma$,
type $v'$ in $\gamma$, child of $v'$ in $\gamma$, child of $u'$ in
$\gamma$, none of the above$\}$.

The proof of correctness and label size now follows by setting $x = O(\log^4
n)$ and the same techniques as in \cite{Alstrup02,Alstrup05}, which is
basically checking the cases of different pairs of node types.

\subsection{Proof of Thorem~\ref{lem:forest-reuse} }  \label{proof:forest-reuse}
    Consider label sets $s_1$ and $s_2$ of two sibling groups from $F_n(a,b)$ and
    $F_n(c,d)$ respectively for which $|s_1\cap s_2| \ge 1$. Clearly, we must have
    $|s_1\cap s_2| \le \min(|s_1|,|s_2|) = \frac{n}{a\cdot b}$. Furthermore, no
    other   sibling group of $F_n(a,b)$ or $F_n(c,d)$ can be assigned  labels from
    $s_1\cup s_2$, as the sibling relationship must be maintained. We can thus
    create a one-to-one matching between the sibling groups of $F_n(a,b)$ and
    $F_n(c,d)$,
    that have labels in common (note that not all sibling groups will
    necessarily be mapped). Bounding the number of common labels thus becomes a
    problem of bounding the size of this matching.
    In order to  maintain the connectivity relation, sibling groups from one component cannot be matched to several  components. Therefore at most $\min(b,d)$
    sibling groups can be shared per component, and at most $\min(a,c)$ components
    can be shared. Combining this gives the final bound of
    $\min(a,c)\cdot\min(b,d)\cdot\frac{n}{a\cdot b}$.

\subsection{Proof of Theorem~\ref{lem:set-reuse-bound} } \label{proof:set-reuse-bound}
    Assume that the encoder has already assigned labels to the set $e_i$.
    The number of distinct labels of $e_i$ is then exactly
    \[
        n - \left|\bigcup_{j=1}^{i-1} (e_j\cap e_i)\right|\ .
    \]
    Since $|A\cup B| \le |A|+|B|$ this is bounded from below by
    \[
        n - \sum_{j=1}^{i-1} |e_j\cap e_i| \ge
        n - \sum_{j=1}^{i-1} \min(a_j,a_i)\cdot \min(b_j,b_i)\cdot
        \frac{n}{a_i\cdot b_i}\ .
    \]
    Here the inequality follows from Lem.~\ref{lem:forest-reuse}
    
\subsection{Proof of Theorem
\ref{thm:con_sib_dyn_lb}}\label{proof:con_sib_dyn_lb}
The key idea is to create a family of forests, such that the non-unique case reduces to the unique case.

\begin{proof}
    Assume w.l.o.g.~that $n$ is a power of $3$.
    Consider the family of $\log_3 n$ forests $F_n(1,n), F_n(3,n/3),
    F_n(3^2,n/3^2), \ldots,$ $F_n(3^{\log_3 n}, 1)$. Since each sibling group of
    the forest $F_n(3^i,n/3^i)$ has exactly one node, we note that no two
    nodes are siblings. Thus each label of the forest has to be unique, since
    we have assumed that a node is sibling to itself. We can thus use
    Lem.~\ref{lem:forest-reuse} as if we were in the unique case for this
    family of forests.

    Let $e_j$ denote the label set assigned by an encoder for
    $F_n(3^j,n/3^j)$. We  assume that the labels are assigned in the order
    $e_0,\ldots,e_{\log_3 n}$. By Lem.~\ref{lem:set-reuse-bound} the number of
    distinct labels introduced when assigning $e_j$ is at least
    \[
        n - n\sum_{i=0}^{j-1} 3^{i-j} > n/2
    \]
    It follows that when labeling each of the $\log_3 n$ forests in the family,
    any encoder must introduce at least $n/2$ distinct labels, i.e. $\Omega(n\log n)$
    distinct labels in total.
    The family consist of forests with no more than $2n$ nodes, which
    concludes the proof.
\end{proof}

\subsection{Proof of Theorem \ref{thm:con-anc}} \label{proof:con-anc}

For integers $n,a,b$ such that $ab \mid n$, let $G_n(a,b)$ be  a forest
consisting of $a$ components consisting each  of $b$ paths of length
$\frac{n}{ab}$ each connected to a root in the component. Each forest in $G_n(a,b)$ consists of at least
$n$ but no more than $2n$ nodes.

The key idea in the proof of Thm.~\ref{thm:static-combo-lb} is the use of
Lem.~\ref{lem:forest-reuse}. Below we show  Lem.~\ref{lem:forest-reuse-con-anc}
which is is analogous to Lem.~\ref{lem:forest-reuse} which derives the proof of  Thm.~\ref{thm:con-anc} similarly.

\begin{lemma}\label{lem:forest-reuse-con-anc}
    Let $G_n(a,b)$ and $G_n(c,d)$ be two forests such that $ab \ge cd$.
    Fix some unique labeling scheme supporting both
    connectivity and ancestry queries, and denote the set of labels assigned to
    $G_n(a,b)$ and $G_n(c,d)$ as $e_1$ and $e_2$ respectively.
	Then
    \[
        |e_1\cap e_2| \le \min(a,c)\cdot\min(b,d)\cdot\frac{n}{a\cdot b}\ .
    \]
\end{lemma}
\begin{proof}
	Let $s_1$ and $s_2$ be the labels assigned to two paths from $G_n(a,b)$ and
	$G_n(a,b)$ respectively for which $s_1 \cap s_2 \neq \emptyset$.
	The number of labels the paths have in common is at most
	$\abs{s_1} = \frac{n}{ab}$.
	Furthermore,
	no other paths from $G_n(a,b)$ or $G_n(c,d)$ can reuse any labels from
	$s_1 \cup s_2$ since the ancestry relation has to be maintained. Therefore
	we can create a one-to-one matching between the paths from $G_n(a,b)$ and
    $G_n(c,d)$,
    which have at least on label in common (note that not all sibling groups will
    necessarily be mapped).

    Bounding the number of common labels thus reduces to
    bounding the size of this matching.
    In order to  maintain the connectivity relation, paths from one
    component cannot be matched to more than one. Therefore at most $\min(b,d)$
    paths can be shared per component, and at most $\min(a,c)$ components
    can be shared. Combining this gives the final bound of
    $\min(a,c)\cdot\min(b,d)\cdot\frac{n}{a\cdot b}$.
\end{proof}

\end{document}